\documentclass[journal]{IEEEtran}
\usepackage[dvipsnames]{xcolor}
\colorlet{LightRubineRed}{RubineRed!70!}
\colorlet{Mycolor1}{green!10!orange!90!}
\definecolor{Mycolor2}{HTML}{00F9DE}
\usepackage{color} 
\usepackage[normalem]{ulem}
\usepackage{amsmath, amsthm, amssymb}
\usepackage{epsfig}
\usepackage{latexsym}
\usepackage{graphicx} 
\usepackage{relsize}
\usepackage{cite}
\usepackage{subfigure}
\usepackage{bm}

\newtheorem{theorem}{Theorem}

\newcommand{\teeabc}{EE}
\newcommand{\rabc}{R}
\newcommand{\rb}{R_b}
\newcommand{\rh}{R_h}
\newcommand{\eabc}{E}
\newcommand{\eb}{E_b}
\newcommand{\eh}{E_h}
\newcommand{\eeabc}{EE}
\newcommand{\Bblb}{B_{\scriptscriptstyle b, LB}}
\newcommand{\Bbub}{B_{\scriptscriptstyle b, UB}}

\ifCLASSINFOpdf
\else
\fi

\begin{document}
\title{Opportunistic Ambient Backscatter Communication in RF-Powered Cognitive Radio Networks}
\author{Rajalekshmi Kishore,~\IEEEmembership{Student Member,~IEEE}, Sanjeev~Gurugopinath,~\IEEEmembership{Member, IEEE},\\Paschalis C.~Sofotasios,~\IEEEmembership{Senior Member,~IEEE}, Sami Muhaidat, \IEEEmembership{Senior Member,~IEEE},\\and Naofal Al-Dhahir,~\IEEEmembership{Fellow,~IEEE}
	\thanks{This work will appear in part in \cite{Kishore_WCNC_2019}.}
	\thanks{R. Kishore is with the Department of Electrical and Electronics Engineering, BITS Pilani, K.~K.~Birla Goa Campus, Goa 403726, India, (email: {\rm lekshminair2k@yahoo.com}).}

	\thanks{S. Gurugopinath is with the Department of Electronics and Communication Engineering, PES University, Bengaluru 560085, India, (email: {\rm sanjeevg@pes.edu}).}

	\thanks{P. C. Sofotasios is with the Department of Electrical and Computer Engineering, Khalifa University of Science and Technology, PO Box 127788, Abu Dhabi, UAE, and with the Department of Electronics and Communications Engineering, Tampere University of Technology, 33101 Tampere, Finland (email: {\rm p.sofotasios@ieee.org}).}

	\thanks{S.  Muhaidat is with the Department of Electrical and Computer Engineering, Khalifa University of Science and Technology, PO Box 127788, Abu Dhabi, UAE and with the Institute for Communication Systems, University of Surrey, GU2 7XH, Guildford, UK, (email: {\rm muhaidat@ieee.org}).}

	\thanks{N. Al-Dhahir is with the Department of Electrical Engineering, University of Texas at Dallas, TX 75080 Dallas, USA (e-mail: {\rm aldhahir@utdallas.edu}).}}


\maketitle

\begin{abstract}
In the present contribution, we propose a novel opportunistic ambient backscatter communication (ABC) framework for radio frequency (RF)-powered cognitive radio (CR) networks. This framework considers  opportunistic spectrum sensing integrated with ABC and harvest-then-transmit (HTT) operation strategies. 
Novel analytic expressions are derived  for the average throughput, the average energy consumption and the energy efficiency in the considered set up. 
These expressions are represented in closed-form and have a tractable algebraic representation which renders them convenient to handle  both analytically and numerically. 
In addition,  we   formulate an optimization problem to maximize the energy efficiency of the CR system operating in mixed ABC $-$ and HTT $-$ modes, for a given set of constraints including   primary interference and   imperfect spectrum sensing constraints. 
Capitalizing on this, we determine the  optimal set of parameters which in turn  comprise the optimal detection threshold, the optimal degree of trade-off between the CR system operating in the ABC $-$ and HTT $-$ modes and the optimal data transmission time. Extensive results from respective computer simulations  are also presented for corroborating  the corresponding analytic results and to demonstrate the performance gain of the proposed model in terms of energy efficiency.
\end{abstract}


\IEEEpeerreviewmaketitle

\section{Introduction}
The need for efficient utilization of spectrum resources has become a fundamental requirement in modern wireless networks due to the witnessed  spectrum scarcity and the ever-increasing demand for higher data rate applications and internet services. 
In this context, an interesting proposal has been the development of cognitive radio (CR) networks \cite{Mitola_IEEE_1999}, which  can adapt their transmission parameters according to the characteristics of the communication environment. 
Cognitive radios have been shown to be efficient in increasing spectrum utilization due to their  inherent spectrum sensing (SS) capability \cite{Axell_SP_2012}. In  this regard, dynamic spectrum access (DSA), where the secondary users (SU) can opportunistically access the underutilized frequency bands, is the standard  solution for the realization of DSA \cite{Altrad_IEEEVT_2014}, which is envisioned to be an integral part of  future  communication systems \cite{Hasan_IEEaccess_2016}. \textcolor{black}{In order to realize DSA, three strategies have been proposed, namely the underlay, the overlay and the interweave. In the underlay technique, the SUs coexist with a PU provided that the interference level at the PU remains below a certain threshold \cite{Usman_IEEESJ_2014}. In the overlay paradigm, the SUs would be allowed to share the band with PU by exploiting the knowledge of its message and codebook in order to reduce interference. Finally, in the interweave technique,  the SU can only access the licensed spectrum of the PU when it is idle \cite{Goldsmith_IEEE_2009}.}

Recently, energy efficiency (EE) has emerged as a major  design and performance criterion for the current and forthcoming wireless systems  \cite{Costa_IEEEJSAC_2016, Zhang_IEEEJSAC_2016, RobatMili_IEEETC_2016, Kishore_Adhoc_2017, Hu_IEEETC_2016}, mainly driven by the ever increasing operating expenditure of  communication networks. 
In this context, it has been shown that combining  effective energy harvesting (EH) techniques with CR can simultaneously improve  the spectrum efficiency and the energy efficiency \cite{Miridakis_IEEEJSAC_2016}, \cite{Shafie_IEEEWC_2015}. Furthermore,  powering  mobile devices by harvested energy  from ambient sources and/or external transmission activities enables   wireless networks to achieve an increased degree of  self-sustainability for a longer period of time \cite{Ren_IEEEcommMag_2018}. 
More recently, the integration of RF energy harvesting techniques with CR networks has lead to the development of a new communication paradigm, known as \emph{RF-powered CR networks} \cite{Kim_IEEETWC_2017}.  In such networks, a CR transmitter harvests RF energy when a primary user is present and utilizes it for   data transmission when the spectrum is vacant. This protocol is referred to as  \emph{harvest-then-transmit} (HTT)   \cite{Kim_IEEETWC_2017, Hoang_IEEE_2017}. 

However, a major challenge associated with this method  is the  reduction of the throughput of the secondary network when the harvested energy is low and/or when the data transmission time is shorter. 
To overcome this shortcoming, the concept of simultaneous wireless information and power transfer (SWIPT) was introduced \cite{Varshney_IEEEConf_2008}, which has attracted significant research attention \cite{Lu_IEEEComsurvey_2015}, \cite{Perera_IEEECST_2018}.
 In addition, in 5G communication systems and beyond, as well is in the context of the Internet of Things (IoT) applications, the SWIPT technology can be  fundamentally important for energy and information transmission across different wireless systems and network architectures, including CR-based networks \cite{Yan_IEEETVT_2017, Benkhelifa_IEEETGC_2017, habob_IEEE_2017}.

\subsection{Ambient Backscatter Communication (ABC)}
	
Ambient backscatter communication (ABC) has recently emerged as a new communication paradigm that is characterized by low power and low cost requirements,  which renders it a strong candidate  for several  IoT based applications \cite{Ensworth_IEEE_2017,Shao_IEEE_2017}. \textcolor{black}{In an ABC system, there are typically two main components, namely, an RF source which acts as a carrier emitter and a backscatter receiver. The ambient RF sources, e.g., TV towers, cellular base stations, and WiFi APs act as carrier emitters. Therefore, the deployment of dedicated RF sources is not required – as opposed to the case of conventional backscattering communication systems. As a result, this reduces the power consumption and overall cost. Secondly, by utilizing existing ambient RF signals, there is no need to allocate new frequency spectrum for ABC, and hence the spectrum utility is improved \cite{Huynh_IEEEST_2018}. In the context of CR, the secondary transmitter (ST) can communicate with a secondary	receiver (SR) by backscattering the primary user (PU) signal, whenever the PU is active. In other words, instead of initiating a CR transmission only when the PU is inactive, the ST can backscatter the PU signal to SR, even when the PU is active. For example, ST can employ the ON-OFF keying strategy to indicate bit \emph{1} or bit \emph{0} by switching its antenna between reflecting or non-reflecting states, respectively.}

Based on the above, it is evident that  the performance of ABC-based CR networks  depends considerably on the availability of PU signals, which represents a major challenge for CR networks particularly during the long idle period. Therefore, this requires a paradigm shift towards the development of key enabling techniques for next generation CR networks, such as the hybrid ABC-HTT schemes, which were recently proposed in  \cite{Hoang_IEEE_2017}. 
However, a common and major drawback in the proposed models  is the assumption of perfect knowledge of PU activities, which is   unrealistic in practical  CR based communication scenarios.  
To this effect,  we propose a novel opportunistic hybrid ABC-HTT model for CR networks, coined as \emph{ABC-HTT-based CR networks}. The proposed framework exploits the potentials of both ABC and RF-powered CRNs; hence, in the context of the proposed framework, we further evaluate and quantify the performance of CR networks by taking into account  the incurred sensing errors under different realistic communication scenarios.

\subsection{Related Work and Motivation}

\subsubsection{RF Powered Cognitive Radio Networks}
RF-based energy harvesting for CR networks is an energy efficient approach to harvest  energy from PU activity  in order to maximize the network capacity 
\cite{Ren_IEEECommMaz_2018}, \cite{Ahmed_IEEETMC_2018}. In \cite{Benkhelifa_IEEETCCN_2017}, the authors investigated SWIPT for spectrum sharing in CR networks, where a CR receiver harvests energy from primary and secondary transmissions using antenna switching. In this work, antennas were selected based on two schemes, namely, the prioritizing data selection (PDS) scheme and  the prioritizing energy selection (PES) scheme.
Then,  a solution was proposed for  the optimal energy-data trade-off study for both PDS and PES schemes under different fading conditions. In \cite{Wang_IEEEaccess_2017}, Wang et al.~introduced a channel access strategy to maximize the sum throughput of secondary users by jointly optimizing the energy harvesting time, resource allocation, and transmit power. Closed-form expressions for the optimal transmit power and channel allocation were also derived, whilst it was shown   that there exists a tradeoff between the sum throughput of CR network and harvested energy.

\subsubsection{Ambient Backscatter Communication}

Recently, Hoang et al.~\cite{Hoang_IEEE_2017} demonstrated  that incorporating ABC with RF-powered CR networks  improves significantly the secondary network throughput. This is because when the primary transmitter is active, the CR transmitter can utilize ABC to transmit its own data to the intended CR receiver. Also, the authors  explored a tradeoff between CR transmission in the ABC and HTT modes,  and provided insights on the optimal time duration in these two modes. 
In \cite{Kim_IEEETWC_2017}, a hybrid backscatter communication for a wireless powered heterogeneous network was introduced, where the  HTT protocol may not be optimal due to the strict energy constraint for active RF communication. 
In addition to HTT, as the primary access protocol, long-range bi-static scatter and short-range ambient backscatter were   adopted in order to increase the  transmission range and provide uniform rate distribution. Likewise, the authors in \cite{Clerckx_IEEEComletter_2017}, \cite{Lu_IEEEWC_2018} focused on the tradeoff between energy harvesting, active transmission and ambient backscatter communication and  demonstrated the superiority of the hybrid scheme in terms of achieved  throughput.
Also, the  effect of physical parameters on the capacity of both legacy and backscatter channels were analyzed  in \cite{Darsena_IEEETC_2017} by considering different receiver architectures. 
Assuming  practical operating conditions, it was shown that  a legacy system employing an orthogonal frequency division multiplexing (OFDM) modulation can turn the RF interference arising from the backscatter process into a form of multipath diversity that can be exploited for enhancing  its performance.
	 
Nevertheless, despite the usefulness of relevant existing contributions, the focus has been largely on the performance enhancement in terms of achievable throughput, and not in terms of energy efficiency and/or minimization of energy consumption, which is a vital metric in the context of the considered system.  
\textcolor{black}{It is recalled here that energy efficiency (EE) is an important performance evaluation metric for a CRN. It is defined as the ratio of the average achievable throughput to the average energy consumption, measured in bits/Hz/J \cite{Althu_dessertation_2014}. It can be noted that the detection accuracy in spectrum sensing affects both the average network throughput and the average energy consumption. However, there exists a tradeoff in optimizing the two metrics, since an increase (or decrease) in average achievable throughput results in an increase (or decrease, respectively) in the average energy consumption. The energy efficiency combines both these metrics, and hence it is capable of accounting  more effectively for the overall performance of a CR system, as a function of the detection accuracy.}

\subsection{Contributions}

Motivated by the above, in the present study  we analyze the energy efficiency (EE) performance of an ABC-HTT-based CR network  in the presence of sensing errors and without assuming knowledge of the PU activity. 
For simplicity, we consider energy detection-based spectrum sensing, which has widely known advantageous characteristics.  
In this context,  we derive analytic expressions for the average achievable throughput and average energy consumption followed by a detailed formulation of  an optimization problem that maximizes the energy efficiency subjected to several constraints,  including the interference constraint on PU. Based on this, we then derive the expressions for the optimal detection threshold,  optimal harvesting time and  optimal data transmission time, and quantify the tradeoff between the ABC and HTT modes, all in terms of energy efficiency.

Specifically,  the main contributions of the present work are listed below:

\begin{itemize}
\item We propose a novel opportunistic ABC framework for RF-powered CR networks in the presence of sensing errors, which operates in combination with the existing HTT mode. We call the proposed network model as \emph{ABC-HTT-based CR network}.
\item We derive novel analytic  expressions for the average achievable throughput, the average energy consumption and the energy efficiency of the proposed ABC-HTT-based CR network.
\item We formulate an optimization problem that maximizes the energy efficiency of the considered network, and evaluate the optimal detection threshold, the optimal energy harvesting time and the optimal data transmission time, subject to PU interference and energy harvesting constraints. 
Also, we quantify the requirements on the backscattering data rate and transmit power of the CR network as well as their impact on finding the optimal energy harvesting time and data transmission time.
\item We present detailed numerical results, which validate our analysis  and evaluate the energy efficiency performance of the CR network in the proposed realistic platform which includes the sensing errors. 
Furthermore, we quantify  the trade-off between   ABC and HTT modes in terms of energy efficiency. It is shown  that operating the CR network in a combination of these two modes improves the overall energy efficiency. In addition, the impact of sensing errors on the overall system performance is  addressed.
\end{itemize}

To the best of the authors' knowledge, no analysis on the energy efficiency in the presence of sensing errors  in the context of ABC-HTT-based CR networks has been reported in the open technical literature.

\subsection{Organization}

The remainder of this paper is organized as follows: Sec.~\ref{SecNetModel} describes the proposed network model and summarizes the performance of energy-based spectrum sensing. The energy efficiency expression is derived in Sec.~\ref{SecEEAnalysis} and the optimization problem is formulated in Sec.~\ref{SecProbForm}. 
The corresponding analysis and related  insights are provided in Sec.~\ref{SecAnalysis}, followed by the corresponding validation through   comparisons with  numerical results in Sec.~\ref{SecNumRes}. Finally, conclusions are drawn in Sec.~\ref{SecConc}.

\section{Network Model} \label{SecNetModel}

Consider an ABC-HTT-based cognitive radio network as shown in Fig.~\ref{FigSysModel}, which consists of a secondary user transceiver pair, denoted by (ST, SR), and a primary transceiver pair, denoted by (PT, PR). We model the CR network in the opportunistic spectrum access (OSA) paradigm, in which the PU channels are accessed opportunistically using SS to detect spectrum holes. The ST is equipped with an energy-based SS unit, an RF energy harvesting unit and an ABC unit. \textcolor{black}{Also, we consider a typical coarse sensing framework, where SS is carried out followed by data transmission over a time frame of $T_{fr}$ seconds, which is normalized such that $T_{fr}=1$.} The time  diagram for the proposed model is shown in Fig.~\ref{FigSysModel}; based on this, when  the PT is declared   present, the ST can harvest energy and store it in a battery, or perform ABC for data transmission, as shown in Fig.~\ref{FigSysModel}(a). 
In this case, the network is   in the \emph{ambient backscatter communication} (ABC) mode, where  $\tau$ denotes the normalized data transmission period, and $(1-\tau)$ denotes the normalized sensing duration of the secondary user transceiver pair. Furthermore, we let $\alpha \tau$ represent  the time fraction utilized for energy harvesting and $(1-\alpha) \tau$ represent  the time fraction for ABC, when the PT-PR channel is declared   occupied. 
The harvested energy during the time $\alpha \tau$ will be stored in the ST battery, and is used for   data transmission over the ST-SR link, when the PT-PR channel is idle.
 On the contrary, when the PT is declared   absent, the ST uses the harvested energy to transmit data to SR during the data transmission period. In this case, the network is said to be in the \emph{harvest-then-transmit} (HTT) mode, which is shown in Fig.~\ref{FigSysModel}(b). In this case, $\mu \in (0,1)$ denotes the fraction of $\tau$ which is used for data transmission, the choice of which depends on the amount of harvested energy.

\begin{figure*}[t]
	\centering
	\includegraphics [width=6.5in, height=2.5in]{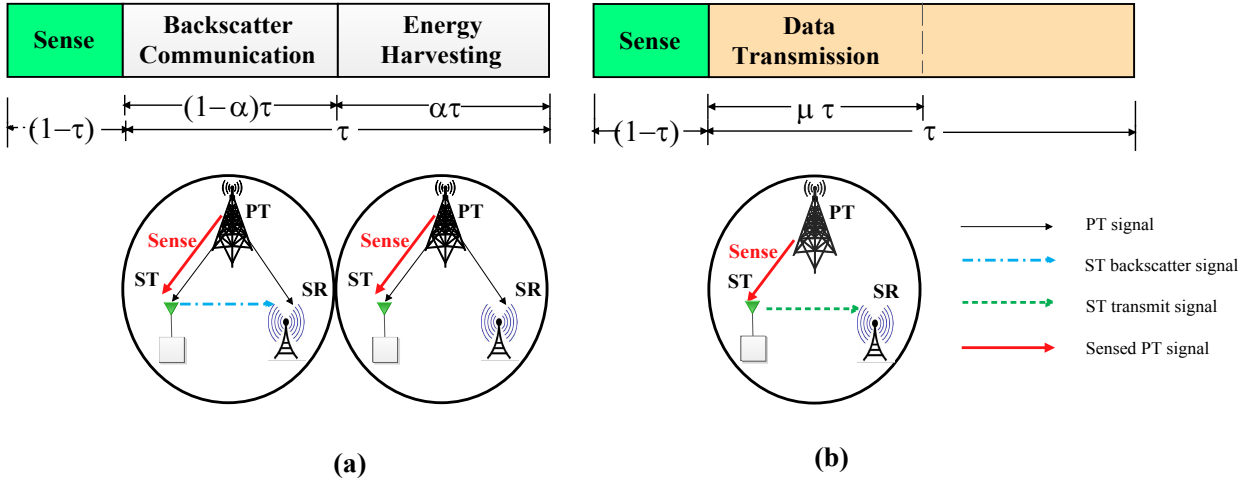}
	\caption{Time slot structure when: (a) the PU is declared to be present; (b)  when the PU is declared to be absent.} \label{FigSysModel}
\end{figure*} 

It is recalled that in the present set up,  we consider energy detection-based SS, given its numerous advantages such as  simple realization and moderate computational complexity \cite{Yucek_IEEE_2009}. 
Based on this, the  probabilities of false-alarm and signal detection at the ST are  given by \cite{Liang_IEEE_2008} 
\begin{align}\label{PfEqn} 
 P_{f}  = Q \left[ \left( \frac {\varepsilon} {\sigma^{2}} -1 \right) \sqrt{ (1-\tau)N_{s}} \right],
\end{align}
and
\begin{align} \label{PDEqn}
 P_{d} = Q \left[ \left(\frac {\varepsilon} {\sigma^{2}}-\gamma -1 \right) \sqrt{\frac{ (1-\tau)N_{s}} {2\gamma+1}} \right],
\end{align}
respectively, where $N_s = f_{s} T_{\scriptstyle t}$ denotes the number of observations, $f_s$ is the sampling frequency, $T_{\scriptstyle t}$ is the duration of the entire frame, $\sigma^2$ is the noise variance, $\varepsilon$ is the detection threshold, and $\gamma$ denotes the received SNR at ST. Furthermore, $Q(\cdot)$ denotes  the complementary cumulative distribution function of a standard Gaussian random variable.

When the PT is declared to be occupied, the ST employs ABC to transmit its own data to the SR, such that a certain quality of service is guaranteed to  the primary system. On the contrary, when the PT is declared to be  inactive, the ST operates in the HTT mode using conventional RF transmission. It is noted here that it was recently  shown that switching between these two modes improves the overall throughput of the secondary system \cite{Hoang_IEEE_2017}. A similar idea is adopted in the present work with the difference that our analysis concerns the study of a CR network operating in ABC-HTT framework, using the OSA paradigm in the presence of the sensing errors  in terms of $P_f$ and $P_d$, as opposed to the analysis  in \cite{Hoang_IEEE_2017} which considers CR operation  in the \textcolor{black}{opportunistic communication mode.} 
In addition, we quantify the  performance  of the proposed model in terms of energy efficiency of the CR network in the presence of sensing errors,  unlike \cite{Hoang_IEEE_2017} that analyzes the throughput performance in the simplistic case of no sensing errors.

\section{Energy Efficiency and Problem Formulation} \label{SecEEAnalysis}

\subsection{Average Achievable Throughput and Energy Consumption} 
It is recalled that the  energy efficiency of the CR network is a defined as the ratio of its average achievable throughput to its average energy consumption \cite{Ejaz1,KoutM}. In what follows,  we calculate the energy efficiency of the proposed model and then  formulate an optimization problem that enables  the calculation of the  optimal values of $\varepsilon$, $\mu$, $\alpha$ and $\tau$ that maximize the energy efficiency, under PU interference and energy harvesting constraints. 
It is noted that due  to the presence of the sensing mechanism, the average achievable throughput depends on the sensing accuracy  and the communication link between PT and ST. This can be categorized into the following four scenarios, which are summarized in Table~\ref{Table1-ABC}. \textcolor{black}{Here, $P(\mathcal{H}_{0})$ and $P(\mathcal{H}_{1})$ denote the prior probabilities of the PT being inactive and active, respectively.}\\

\begin{table*}[t!]  \centering {
		\centering
		\label{TputEconTable}
		\scalebox{1}{
		} }
	\end{table*}
	
	\begin{table*}[t!]
		\centering
		\caption{Average achievable throughput and the average energy consumption for different communication scenarios in ABC-HTT-based CR network.}
		\label{Table1-ABC}
		\begin{tabular}{|c|c|c|c|l}
			\cline{1-4}
			Scenarios            & Harvested Energy       & Consumed Energy               & Throughput                                                                                                                   &  \\ \cline{1-4}
			$P(\mathcal{H}_{1})P_{d}$     & $\alpha$ $\tau$ $P_{R}$ & $P_{s}(1-\tau)$               & $R_{b}=(1-\alpha) \tau B_{b}$                                                                                                      &  \\ \cline{1-4}
			$P(\mathcal{H}_{0})P_{f}$     & 0                       & $P_{s}(1-\tau)$               & 0                                                                                                                            &  \\ \cline{1-4}
			$P(\mathcal{H}_{1})(1-P_{d})$ & 0                       & $P_{s}(1-\tau)+P_{tr}\mu \tau$ & \begin{tabular}[c]{@{}c@{}}$R_{h}=\kappa \mu \tau W \log_{2}(1+\frac{P_{tr}}{Z_{I}P_{T,PU}+P_{0}}), ~~\kappa \in (0,1)$ \end{tabular} &  \\ \cline{1-4}
			$P(\mathcal{H}_{0})(1-P_{f})$ & 0                       & $P_{s}(1-\tau)+P_{\mathrm{tr}}\mu \tau$ & $R_{h}=\mu \tau  W \log_{2}(1+\frac{P_{\mathrm{tr}}}{P_{0}})$                                                                          &  \\ \cline{1-4}
		\end{tabular}
	\end{table*}

\textbf{S1:} In this scenario, the ST correctly declares  the presence of the PT  with probability $P(\mathcal{H}_{1}) P_{d}$. Since the licensed band is occupied and the primary transmission is active, the throughput in this case is achieved due to ST using only the ABC mode, and \textcolor{black}{is given by \cite{Hoang_IEEE_2017}}
\begin{align}
R_{b, \scriptstyle S_1}=(1-\alpha) \tau B_{b},
\end{align}
where $B_{b}$ is the achievable backscatter rate in the ABC mode.

\textcolor{black}{It is worth noting that in this scenario, the SR should be able to decode the data without using power-demanding components such as analog to digital converter (ADC) and oscillators. An ultra low power receiver should be utilized to decode the modulated signal \cite{vincent_ACM_2013}. The receiver strategy proposed in \cite{vincent_ACM_2013}, namely the averaging mechanism, requires only an envelope average and threshold calculator. The envelope circuit first smoothes the average of the received signals, and then a threshold value based on two signal levels is calculated. Finally, the smoothened signal strengths are compared with this selected threshold to detect bits 1 and 0, followed by decoding.}  

\textbf{S2:} In this scenario, the ST incorrectly declares the PT to be active with probability $P(\mathcal{H}_{0}) P_{f}$. This results to a lack of throughput  since the CR network achieves no throughput by operating in the ABC mode. Furthermore, ST misses a transmission opportunity.

\textbf{S3:} In this scenario, the ST incorrectly declares  the PT to be absent, with probability $P(\mathcal{H}_{1}) (1-P_{d})$; as a consequence,  the ST misses an opportunity to use the ABC mode. Moreover, the ST transmits to SR in the HTT mode and creates  interference to PT. In the presence of the interference from PT, the CR network achieves a partial throughput of
\begin{align}
R_{h, \scriptstyle S_3}=\mu \tau \kappa W \log_{2}\left(1+\frac{P_{\mathrm{tr}}}{Z_{I}P_{T,PU}+P_0}\right),
\end{align}
\textcolor{black}{with a partial throughput factor $\kappa \in (0,1)$, which quantifies the partial throughput achievable in this scenario}, where  $W$ is the bandwidth of the primary link, $P_{0}$ is the ratio between the noise power $N_{0}$ and $g_{c}$, the channel gain coefficient between ST and SR, that is $P_{0}=\frac{N_{0}}{g_{c}}$, $P_{\mathrm{tr}}$ denotes the transmit power of the ST in the data transmission period $\mu \tau \in (0,\tau)$, as shown in Fig.~\ref{FigSysModel}(b), $Z_{I}$ denotes the ratio of the channel gain between the PT and the ST to $g_c$, and $P_{T, \rm PU}$ denotes the transmission power of the PU.

Next, $P_{\mathrm{tr}}$ can be expressed as 
\begin{equation}
P_{\mathrm{tr}} = \frac {E_{\mathrm {h}} -E_{\mathrm {s}}- E_{\mathrm {c}}}{\mu \tau },
\label{ptrequation}
\end{equation}
where $E_{\mathrm {s}}=P_{s}(1-\tau)$ is the energy consumed during sensing, $E_{\mathrm {c}}=\mu \tau P_c$ is the energy consumption of the circuitry in the transmission time $\mu \tau$, $E_{\mathrm {h}}=\alpha \tau P_{R}$ is the total harvested energy, and $P_{R}$ is the harvested RF power obtained from the PT signal at the ST, which is determined from the Friis' equation as follows \cite{Balanis_Wiely_2005} :
\begin{equation}
P_R = \delta P_{T, PU} \frac{G_T G_R \lambda^2}{(4 \pi d)^2},
\end{equation}
where $\delta \in [0,1]$ is the energy harvesting efficiency, $G_{T}$ is the PT antenna gain, $G_{R}$ is the ST antenna gain, $\lambda$ is the wavelength of the emitted wave, and $d$ is the distance between the PT and ST. Based on the above, it follows that 
\begin{align}
R_{h, \scriptstyle S_3}=\mu \tau \kappa W \log_{2}\left(1+\frac{\alpha \tau P_{R}-\mu \tau P_{c}-P_{s}(1-\tau)}{\left[Z_{I}P_{T,PU}+P_{0}\right]\tau \mu}\right). 
\end{align}

\textbf{S4:} In this scenario, the ST correctly declares  the PT to be inactive with probability $P(\mathcal{H}_{0})(1-P_{f})$, and hence the CR network achieves the maximum achievable throughput in the HTT mode, namely 
\begin{align}
R_{h, \scriptstyle S_4}=\mu \tau  W \log_{2}\left(1+\frac{P_{tr}}{P_{0}}\right).
\label{Rhs4}
\end{align}
Hence, substituting \eqref{ptrequation} into \eqref{Rhs4} yields 
\begin{align}
R_{h, \scriptstyle S_4}=\mu \tau  W \log_{2}\left(1+\frac{\alpha \tau P_{R}-\mu \tau P_{c}-P_{s}(1-\tau)}{P_{0}}\right).
\end{align}

Considering the above four scenarios, the average throughput of the ABC-HTT-based CR network is given by 
\begin{align} \nonumber
& \rabc(\tau, \alpha,\mu, \varepsilon) = P(\mathcal{H}_{1})P_{d} (1 \hspace{-0.1cm} - \hspace{-0.1cm} \alpha) \tau B_{b} \hspace{-0.1cm} + \hspace{-0.1cm}  \kappa P(\mathcal{H}_{1})(1-P_{d}) \nonumber \\
& ~~~~~~~~~~~~~~~~~ \mu \tau  W \hspace{-0.05cm} \log_{2} \hspace{-0.1cm} \left( \hspace{-0.1cm} 1 \hspace{-0.1cm} +  \hspace{-0.1cm} \frac{P_{tr}}{Z_{I}P_{T,PU}+ P_{0}} \hspace{-0.1cm}\right) \nonumber \\ 
&~~~~~~~~~~~~~~~~~~  + P(\mathcal{H}_{0})(1-P_{f}) \mu \tau  W \hspace{-0.05cm} \log_{2} \hspace{-0.1cm} \left( \hspace{-0.1cm} 1 \hspace{-0.1cm} +  \hspace{-0.1cm} \frac{P_{tr}}{P_{0}} \hspace{-0.1cm}\right),
\label{Ravg}
\end{align}
which can be equivalently expressed as 
\begin{align} 
& \rabc(\tau, \alpha,\mu, \varepsilon) = \rb(\tau, \alpha,\mu, \varepsilon) + \rh(\tau, \alpha,\mu, \varepsilon),
\end{align}
where $\rb(\tau, \alpha,\mu, \varepsilon)$ denotes the average achievable throughput of the CR network in the ABC mode, given by
\begin{align} 
\rb(\tau, \alpha,\mu, \varepsilon) \triangleq P(\mathcal{H}_{1})P_{d} (1 - \alpha) \tau B_{b},
\end{align}
and $\rh(\tau, \alpha,\mu, \varepsilon)$ denotes the average achievable throughput of the CR network in the HTT mode, given by
\begin{align} 
& \rh(\tau, \alpha,\mu, \varepsilon)\nonumber \triangleq   \kappa P(\mathcal{H}_{1})(1-P_{d}) \nonumber \\
& ~~~~~~~~~~~~~~~ \mu \tau  W \hspace{-0.05cm} \log_{2} \hspace{-0.1cm} \left( \hspace{-0.1cm} 1 \hspace{-0.1cm} +  \hspace{-0.1cm} \frac{P_{tr}}{Z_{I}P_{T,PU}+ P_{0}} \hspace{-0.1cm}\right)\nonumber \\ 
& ~~~~~~~~~~~~~~~ + P(\mathcal{H}_{0})(1-P_{f}) \mu \tau  W \hspace{-0.05cm} \log_{2} \hspace{-0.1cm} \left( \hspace{-0.1cm} 1 \hspace{-0.1cm} +  \hspace{-0.1cm} \frac{P_{tr}}{P_{0}} \hspace{-0.1cm}\right).
\end{align}

It is noted that in order for the throughput to be non-negative, the harvested energy should be greater than the consumed energy. Thus, this requirement imposes the following constraint
\begin{align}
E_{h} = \alpha \tau  P_{R}\geq E_{c} +E_{s},
\end{align}
which implies that 
\begin{align}
\alpha \geq \frac{E_{c}+E_{s}}{\tau \ P_{R}}.
\end{align}
Denoting $\alpha^{\dagger} \triangleq (E_{c}+E_{s})/(\tau \ P_{R})$ as the minimum energy harvesting time to obtain enough energy for the ST to operate in the HTT mode, we have the constraint that $\alpha \in [\alpha^{\dagger},1]$. In other words, $\rh(\tau, \alpha,\mu, \varepsilon) > 0$, only when $\alpha \in [\alpha^\dagger, 1]$; otherwise, $\rh(\tau, \alpha,\mu, \varepsilon) = 0$, that is
\begin{align}
& \rh(\tau, \alpha,\mu, \varepsilon) \nonumber \\
& ~ = \begin{cases}
 \kappa P(\mathcal{H}_{1})(1-P_{d}) \mu \tau  W \hspace{-0.05cm} \log_{2} \hspace{-0.1cm} \left( \hspace{-0.1cm} 1 \hspace{-0.1cm} +  \hspace{-0.1cm} \frac{P_{tr}}{Z_{I}P_{T,PU}+ P_{0}} \hspace{-0.1cm}\right) \\ 
~~ + P(\mathcal{H}_{0})(1-P_{f}) \mu \tau  W \hspace{-0.05cm} \log_{2} \hspace{-0.1cm} \left( \hspace{-0.1cm} 1 \hspace{-0.1cm} +  \hspace{-0.1cm} \frac{P_{tr}}{P_{0}} \hspace{-0.1cm}\right) ,  ~  \text{ if } \alpha^{\dagger}\leq \alpha \leq 1, \\ 
0, ~~~~~~~~~~~~~~~~~~~~~~~~~~~~~~~~~~~~~~~~~~~~~~  \text{otherwise.} 
\end{cases}
\end{align}
By recalling that $P_{s}$ denotes the power required by ST to perform sensing,  the average energy consumption in the CR network, from Table~\ref{Table1-ABC}, is given by 
\begin{align} 
& \eabc(\tau, \alpha,\mu, \varepsilon) = \eb(\tau, \alpha,\mu, \varepsilon)+\eh(\tau, \alpha,\mu, \varepsilon)  \\
& \phantom{\eabc(\tau, \alpha,\mu, \varepsilon)} = P_{s}(1-\tau) + \mu \tau P_{\mathrm{tr}}   \{ P(\mathcal{H}_{1})(1-P_{d}) \nonumber \\
& ~~~~~~~~~~~~~~~~~~~~~ +P(\mathcal{H}_{0})(1-P_{f}) \},
\end{align}
where $\eb(\tau, \alpha,\mu, \varepsilon)$ and $\eh(\tau, \alpha,\mu, \varepsilon)$ denote the energy consumed by the CR network, while operating in the ABC mode and HTT mode, respectively.

In the same context, the energy efficiency of the CR network, in bits/Hz/J, is defined as
\begin{align}
& \eeabc(\tau, \alpha,\mu, \varepsilon) \triangleq \frac{\rabc(\tau, \alpha,\mu, \varepsilon)}{\eabc(\tau, \alpha,\mu, \varepsilon)}, \nonumber \\
& \phantom{\eeabc(\tau, \alpha,\mu, \varepsilon)}  = EE_{b}(\tau, \alpha,\mu,\varepsilon)+EE_{h}(\tau, \alpha,\mu,\varepsilon),
\end{align}
where
\begin{align}
 EE_{b}(\tau, \alpha,\mu,\varepsilon) \triangleq \frac{\rb(\tau, \alpha,\mu, \varepsilon)}{\eabc(\tau, \alpha,\mu, \varepsilon)}
 \end{align}
 and
 \begin{align}
 EE_{h}(\tau, \alpha,\mu,\varepsilon) \triangleq \frac{\rh(\tau, \alpha,\mu, \varepsilon)}{\eabc(\tau, \alpha,\mu, \varepsilon)},
\end{align}
denote the energy efficiency values due to the ST operating in ABC and HTT modes, respectively. 
Furthermore,  the constraint $\alpha \in [\alpha^\dagger, 1]$ yields the following condition on the overall energy efficiency.
\begin{align}
& \eeabc(\tau, \alpha,\mu, \varepsilon) \nonumber \\
& ~ = \begin{cases}
\hspace{-.05cm} EE_{b}(\tau, \alpha,\mu,\varepsilon)  +EE_{h}(\tau, \alpha,\mu,\varepsilon), ~~~~ \text{ if } \alpha^{\dagger}\leq \alpha \leq 1, \\
EE_{b}(\tau, \alpha,\mu,\varepsilon), ~~~~~~~~~~~~~~~~~~~~~~~~~~  \text{otherwise.} 
\end{cases}
\end{align}

\subsection{Problem Formulation: Energy Efficiency Maximization} \label{SecProbForm}
In what follows, we describe an optimization problem in order to determine the optimal values of the parameters $\varepsilon$, $\mu$, $\alpha$ and $\tau$, such that the energy efficiency of the CR network is maximized. To this end, we formulate the following maximization problem, subject to the interference constraint on the primary network and energy harvesting constraint.
\begin{align} 
& \mathcal{OP}:  \underset{\tau,\mu,\alpha,\varepsilon}{\max} ~~~ EE(\tau,\alpha,\mu,\varepsilon) \nonumber  \\
& ~~~~~~~~~ \text{s.t.} ~~ P_{f} \leq \overline{P}_{f}, ~~~ \text{for some } \overline{P}_f \in (0,1) \nonumber \\
& ~~~~~~~~~~~~~~~ P_{d} \geq \overline{P}_{d}, ~~~ \text{for some } \overline{P}_d \in (0,1)\nonumber \\
& ~~~~~~~~~~~~~~~ \alpha^\dagger \leq \alpha \leq 1, \nonumber \\
& ~~~~~~~~~~~~~~~ 0 \leq \mu \leq 1, \nonumber \\
& ~~~~~~~~~~~~~~~ 0 \leq \tau \leq 1. \label{Opabc}
\end{align}
In the next section, we provide the detailed solution of the above optimization  problem.

\section{Performance and Energy Efficiency Optimization}  \label{SecAnalysis}
In the following theorem, we derive the optimal value of the detection threshold, $\varepsilon^*$, that satisfies the primary interference constraint given in problem $\mathcal{OP}$.
\begin{theorem} \label{etathm}
	The optimal threshold $\varepsilon^*$ for the problem in $\mathcal{OP}$ is obtained when the constraint $P_d \geq \overline{P}_d$ is satisfied with equality, namely 
	\begin{equation} \label{optthr}
	\varepsilon^{*} \hspace{-0.1cm} = \hspace{-0.07cm} \sigma^2\left [ \left ( \gamma \hspace{-0.1cm} + \hspace{-0.1cm} 1 \right ) +\sqrt{\frac{2\gamma+1}{(1-\tau) N_s}}Q^{-1}\left [\overline{P}_d \right ]\right ].
	\end{equation}
\end{theorem}
\begin{proof}
	The proof is provided in Appendix \ref{etathmProof}.
\end{proof}

In what follows,  Theorem~\ref{formu}  shows that when the spectrum is sensed to be idle, the energy efficiency is maximized when the ST transmits for the entire data transmission period,  that is, when $\mu = 1$.\footnote{It is   worth noting that in  the energy efficiency equation,   only $EE_h(\tau,\alpha,\mu,\varepsilon^*)$ depends on $\mu$.} 
\begin{theorem} \label{formu}
When $\alpha^{\dagger}\leq 1$  and $\alpha \geq \alpha^{\dagger}$, the energy efficiency due to the harvest-and-transmit mode, i.e., $EE_{h}(\tau,\alpha,\mu, \varepsilon^*)$ is  maximum for $\mu^*=1$.
\end{theorem}
\begin{proof}
The proof is provided in Appendix \ref{formuProof}.
\end{proof}

In the same context,  Theorem~\ref{alpha} allows the determination of the  conditions on the backscattering communication rate, such that an optimal value of $\alpha$, denoted by $
\alpha^*$, exists between $\alpha^\dagger$ and $1$, and provides  an analytic expression for $\alpha^*$, when the interference from the PU is neglected. The existence of $\alpha^*$ can be determined similarly for the case that includes the interference term, but it yields intractable results since the derivation of a  closed form solution is infeasible.
\begin{theorem} \label{alpha}
When $\alpha^\dagger \leq \alpha \leq 1$ and the backscatter transmission rate $B_{b} \in \left(\Bblb, \Bbub\right)$, where
\begin{align}
& \Bblb \triangleq \left(\frac{P(\mathcal{H}_{0})}{P(\mathcal{H}_{1})}\right) \left(\frac{(1-P_f)}{P_d \ln 2}\right) \nonumber \\
& ~~~~~~~~~~~~~ \times \left( \frac{\mu \tau W \tau P_R}{(P_0 -P_c )\mu \tau +P_s(1-\tau)+ \tau P_R} \right),
\end{align}
and
\begin{align}
& \Bbub \triangleq \left(\frac{P(\mathcal{H}_{0})}{P(\mathcal{H}_{1})}\right) \left(\frac{(1-P_f)}{P_d \ln 2}\right) \nonumber \\
& ~~~~~~~~~~~ \times \left( \frac{\mu \tau W \tau P_R}{(P_0 -P_c )\mu \tau +P_s(1-\tau)+ \alpha^\dagger \tau P_R} \right),
\end{align}
and the interference from the PU is neglected, then, there exists an optimal solution $\alpha^* \in [\alpha^{\dagger}, 1]$, given by
\begin{align}
& \alpha^*= \left(\frac{P(\mathcal{H}_{0}) }{P(\mathcal{H}_{1})}\right) \left( \frac{(1-P_f)}{P_d} \right) \left( \frac{\mu \tau W}{\ln 2}\right) \nonumber \\
& ~~~~~~~ - \left(\frac{(P_{0} +P_{c})\tau \mu+P_s(1-\tau)}{\tau P_R} \right).
\label{alphaopt}
\end{align}
\end{theorem}
\begin{proof}
The proof is provided in Appendix \ref{alphaProof}.
\end{proof}

Once the optimal values $\varepsilon^*$, $\mu^*$ and $\alpha^*$ are determined, we need to determine  the optimal value of $\tau$, denoted by $\tau^*$, which accounts for the data transmission duration. In the following theorem, we show that the function $\teeabc(\tau, \alpha^*,  \mu^*, \varepsilon^*)$ is concave in $\tau$, and therefore, $\tau^*$ can be determined by standard methods, such as steepest gradient techniques.
\begin{theorem} \label{fortau} 
The function $\teeabc(\tau, \alpha^*, \mu^*, \epsilon^*)$ is concave in $\tau$.
\end{theorem}
\begin{proof}
	The proof is provided in  Appendix \ref{fortauProof}.
\end{proof}

Finally, the maximum energy efficiency can be evaluated as follows: 
\begin{align}
	&EE_{\max}(\tau^*, \alpha^*,\mu^*, \varepsilon^*) \nonumber \\
	&  = \begin{cases}
	\max \left[EE_{b}(\tau^*, 0, \mu^*, \varepsilon^*), \right. \\
	\left. ~EE_{b}(\tau^*, \alpha^*,\mu^*,\varepsilon^*)+EE_{h}(\tau^*, \alpha^*,\mu^*,\varepsilon^*)\right],\\~~~~~~~~~~~~~~~~~~~~~~~ ~~~~~~~~~~~~~~~~~~~~~~ \text{ if } \alpha^{\dagger}\leq \alpha^* \leq 1, \\ 
	EE_{b}(\tau^*,0,\mu^*,\varepsilon^*), ~~~~~~~~~~~~~~~~~~~~~~~\text{otherwise.} 
	\end{cases} \hspace{-0.5cm} \label{EEfinal}
	\end{align}

\section{Numerical Results} \label{SecNumRes}

\begin{figure}[t]
	\centering
	\includegraphics [width=3.5in]{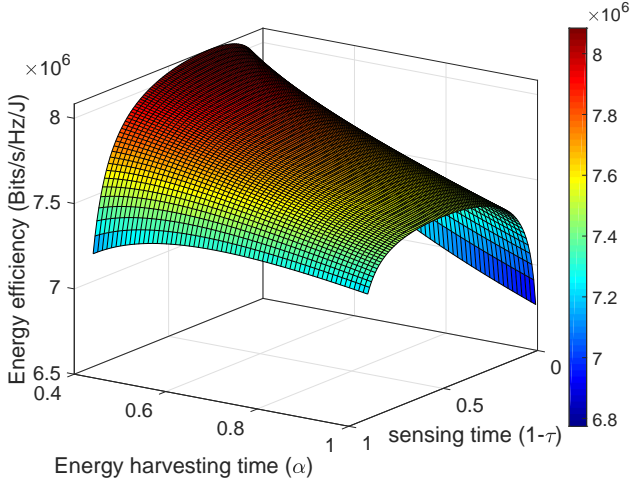}
	\caption{Variation of energy efficiency with $\alpha$ and $\tau$, for $\mu^*=1$ and $\varepsilon^*$ in \eqref{optthr}. Energy efficiency is concave in both $\alpha$ and $\tau$.} \label{3DFig}
\end{figure}

\begin{figure}[t]
	\centering
	\includegraphics[width=3.5in]{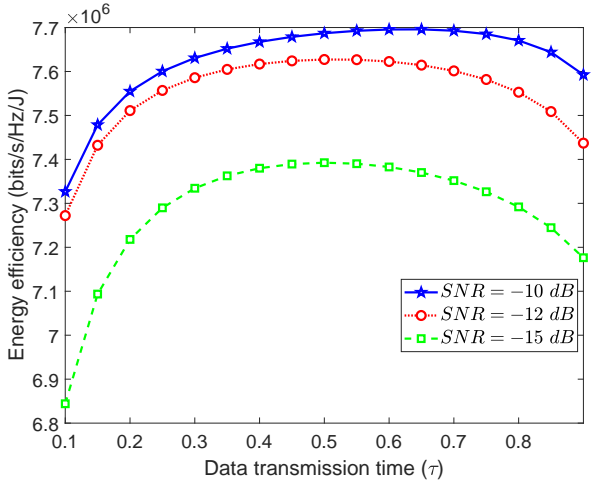}
	\caption{Variation of the energy efficiency $\teeabc(\tau, \alpha^*, \mu^*, \varepsilon^*)$ with the data transmission time, $\tau$, for different SNR values.} \label{FigEEvsTau_SNR}
\end{figure}

\begin{figure}[t]
	\centering
	\includegraphics [width=3.5in]{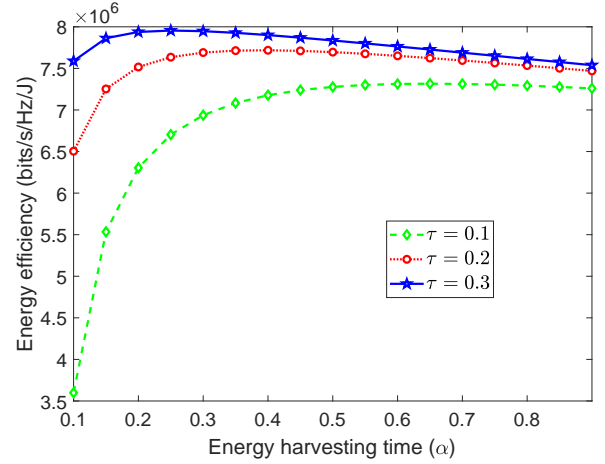}
	\caption{Variation of the energy efficiency $\teeabc(\tau, \alpha^*, \mu^*, \varepsilon^*)$ with $\alpha$, for different values of $\tau$.} \label{FigEEvsAlpha}
\end{figure}

In this section, we present the numerical results on the performance of the ABC-HTT-based CR network. To this end, we consider the following parameters: the target probability of detection, $\overline{P}_{d}$, and false-alarm probability, $\overline{P}_{f}$, are set to be $0.9$ and $0.1$, respectively \cite{Liang_IEEE_2008}, whereas the prior probabilities $P(\mathcal{H}_0)$ and $P(\mathcal{H}_1)$ are set to $0.75$ and $0.25$, respectively. The signal bandwidth and the transmitted power are set to be $6$ MHz, and $17$ kW, respectively \cite{Hoang_IEEE_2017}. 
Also, without a loss of generality and unless stated otherwise, we assume the following values:  The number of observations is $2000$, $B_b=50 \times 10^3$ bps, SNR = $-10$ dB, $\kappa=1$, $P_s = 1$ mW, $P_c = 0.1$ mW, $\delta=0.6$, \textcolor{black}{$d=2.475$ km, $G_T = G_R = 6$ dBi in the Friis' equation, such that $P_R = 0.25$ W}, and the path loss and other impairments due to primary interference $X_l = 0.5 \times 10^{-3}$.

Figure \ref{3DFig} shows the variation of the energy efficiency with respect to the parameters $\alpha$ and $\tau$, with $\mu^*=1$ and $\varepsilon^*$ chosen according to  \eqref{optthr}. The sampling frequency is chosen such that the number of samples is $1000$, and the sensing power $P_{s}=0.3$ mW.  Also, we set the partial throughput factor, $\kappa = 0.6$, and the energy harvesting efficiency, $\delta = 0.6$. It is evident that the energy efficiency is concave with respect to both $\alpha$ and $\tau$. Also, for a small value of $\alpha$, the energy efficiency is small since the throughput decreases due to little energy harvesting. However, if ST spends more time on energy harvesting, i.e., if $\alpha$ increases, the energy efficiency  decreases  further since the backscattering communication is not efficiently utilized.

\begin{figure}[t]
	\centering
	\includegraphics [width=3.5in]{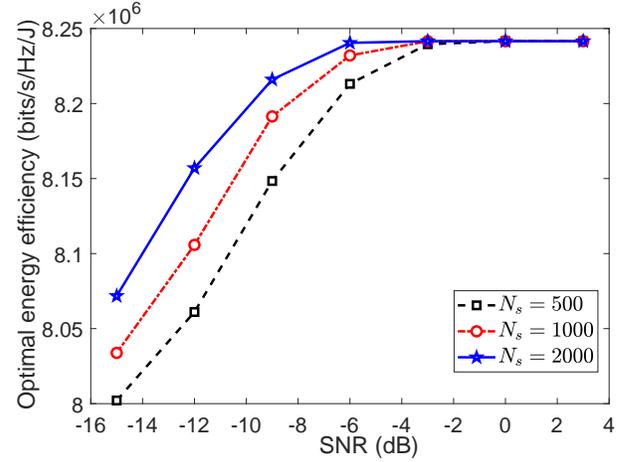}
	\caption{Variation of the optimal energy efficiency $\teeabc(\tau^*, \alpha^*, \mu^*, \varepsilon^*)$ with SNR, for different values of number of samples $N_s$. An increase in the sampling frequency for a fixed $\tau$ increases the number of samples.}
	\label{eevssnrvsfs}
\end{figure}

\begin{figure}
	\centering
	\includegraphics [width=3.5in]{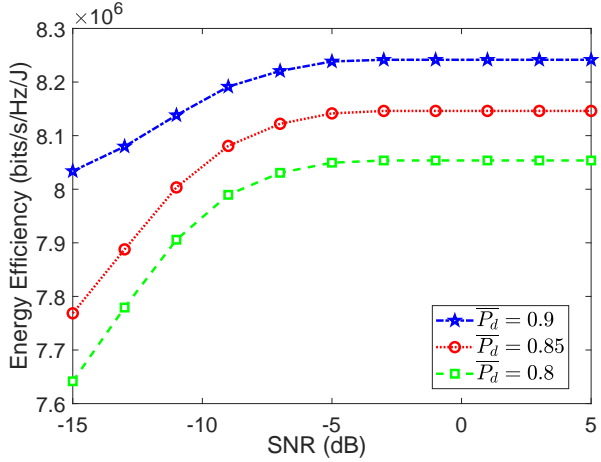}
	\caption{Variation of the optimal energy efficiency $\teeabc(\tau^*, \alpha^*, \mu^*, \varepsilon^*)$ with SNR, for different values of $\overline{P}_d$.}
	\label{EEsnrpdbar}
\end{figure}
   
\begin{figure}
	\centering
	\includegraphics [width=3.5in]{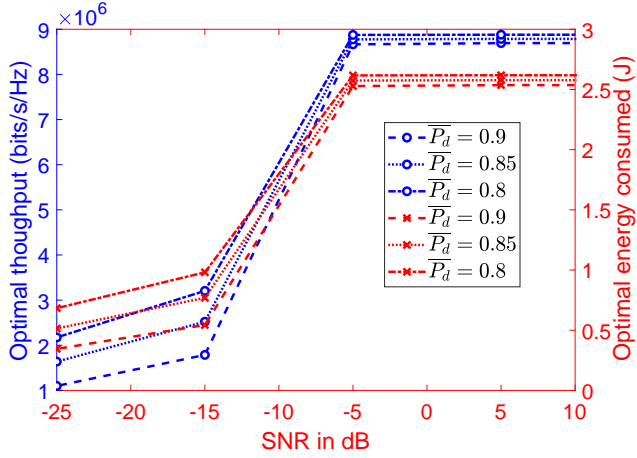}
	\caption{Variation of the optimal achievable throughput and optimal average energy consumption with SNR, for different values of $\overline{P}_d$. Although a relaxed constraint with $\overline{P}_d$  improves the achievable throughput, it also allows an increase in the energy consumption due to more transmission opportunities, which is   significant at low SNR values.}
	\label{thECsnrpdbar}
\end{figure}   

Figure \ref{FigEEvsTau_SNR} illustrates  the variation of the value of $\teeabc(\tau, \alpha^*, \mu^*, \varepsilon^*)$ with $\tau$ for different SNR values. As expected, the energy efficiency increases with SNR. Moreover, the value of optimum $\tau$ also increases with SNR, since a larger SNR results in lower sensing time required to satisfy the primary interference constraints and thus, to a  higher data transmission time. Similarly,  Figure~\ref{FigEEvsAlpha} demonstrates  the variation of the value of $\teeabc(\tau, \alpha, \mu^*, \varepsilon^*)$ for different values of $\tau$. 
It is evident  that the optimal $\alpha^*$ exists for each $\tau$, and it decreases with an increase in $\tau$, which is clear from \eqref{alphaopt}. Moreover, it is intuitive to note that as $\tau$ increases, the energy efficiency increases.

Figure \ref{eevssnrvsfs} illustrates the variation of the optimal energy efficiency for different values of the received SNR at ST, for a number of samples $N_s$= $500$, $N_s$ =  $1000$ and $N_s$ = $2000$. It is shown that as the number of samples increases, the detection accuracy increases, which improves the secondary throughput and the energy efficiency. Likewise,   Figure~\ref{EEsnrpdbar} presents the variation of the optimal energy efficiency with respect to SNR for different target probability of detection values, $\overline{P}_{d}$. It is noted here that  the performance is expected to increase with a decrease in $\overline{P}_{d}$, since a lower tolerance on the probability of detection improves the average throughput and energy consumption. However, the plots exhibit a different trend. 
For example,  Figure~\ref{thECsnrpdbar} reveals that a lower value of $\overline{P}_d$ implies that the PU interference constraint is more relaxed, which in turn yields a better throughput. However, since a lower value of $\overline{P}_d$ also results in more transmission opportunities, the average energy consumption also increases, which is significant at the low SNR regime. Therefore, the trend on variation of EE with $\overline{P}_d$ depends largely  on the choice of system parameters, which explains the trend observed in Figure~\ref{EEsnrpdbar}. Also, the optimal energy efficiency saturates after a certain SNR, since a further improvement in SNR will only improve the detection performance by a small margin, resulting in little improvement in the overall energy efficiency.

\begin{figure}
	\centering
	\includegraphics [width=3.8in]{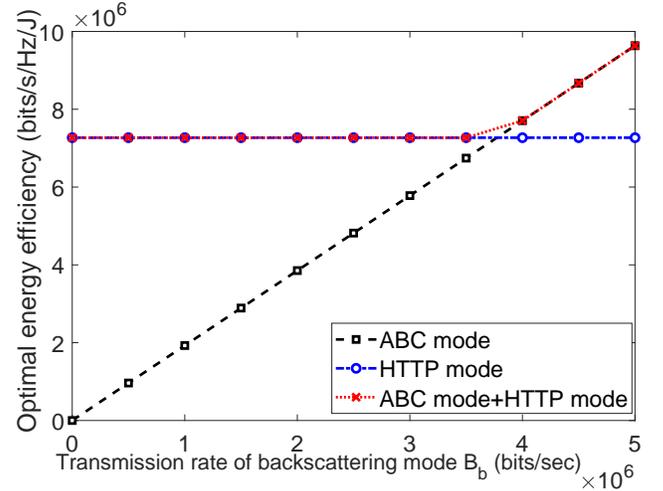}
	\caption{Variation of the optimal energy efficiency $\teeabc(\tau^*, \alpha^*, \mu^*, \varepsilon^*)$ with different values of ABC rates, $B_b$. Using only the HTT mode   yields better performance for small values of $B_b$. Conversely,   using only the ABC mode performs  better for larger values of $B_b$. Therefore, combining the ABC and HTT modes yields an overall better performance in terms of energy efficiency, across all values of $B_b$.}
	\label{ABCHTTmodes}
\end{figure}   

\begin{figure}
	\centering
	\includegraphics [width=3.7in]{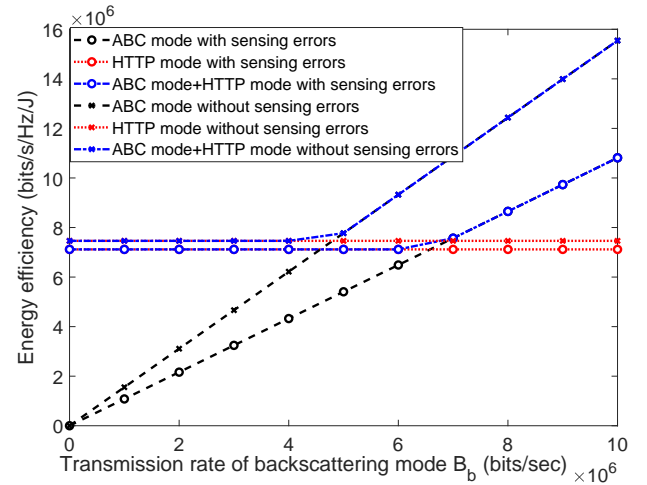}
	\caption{Effect of sensing errors on the optimal energy efficiency. Combining ABC and HTT modes yields a higher overall energy efficiency.}
	\label{EEwithwithoutError}
\end{figure}

Figure  \ref{ABCHTTmodes} demonstrates  the variation of the optimal energy efficiency $\teeabc(\tau^*, \alpha^*, \mu^*, \varepsilon^*)$ with different values of backscattering communication rates, $B_b$ for the indicative values of  $N_s = 1000$, and $P_R = 1$ W. It is evident that  the achievable rate  due to energy harvesting is not dependent upon  $B_b$. Also, for lower values of $B_b$, using the HTT mode alone yields a better energy efficiency,  whereas for larger values of $B_b$, the ABC mode exhibits  a better performance. Therefore, operating the CR network in a combination of the ABC and HTT modes yields an improvement in the overall performance in terms of energy efficiency, across all values of $B_b$.

Finally, Figure~\ref{EEwithwithoutError} illustrates  the effect of sensing errors on the performance of the ABC-HTT-based CR network. 
In order to calculate the performance of the system without sensing errors, we  follow the procedure described in \cite{Hoang_IEEE_2017}, which considers the simplistic case of no sensing errors. By choosing the indicative values  $N_s = 2000$, and $P_R = 1$ W, it is shown  that the optimal energy efficiency achieved with no sensing errors is, as expected, higher  than  the realistic case with present sensing errors. Additionally, we observe that the energy efficiency increases in both cases, due to the use of both ABC and HTT modes. That is, as expected, the energy efficiency achieved due to only ABC or HTT mode is lower than that obtained by combining the two modes, in the presence and absence of sensing errors. \textcolor{black}{Moreover, it is recalled that we have set $\kappa=1$. In terms of energy efficiency, this corresponds to the best possible case from the CR network point of view. In our formulation, the choice of $\kappa$ only affects the average achievable throughput, and not the energy consumption. Therefore, the energy efficiency performance deviation in Fig.~\ref{EEwithwithoutError} -- between the ABC and HTT modes with and without sensing errors -- constitutes a lower bound. In fact, choosing any other value of $\kappa$ will result in a larger performance deviation.}

\section{Conclusion}  \label{SecConc}
We investigated  the performance of ABC-HTT-based cognitive radio networks in terms of energy efficiency in the presence of sensing errors as they are encountered in realistic wireless communication scenarios. In this context, we derived novel analytic expressions for the average achievable throughput, average energy consumption and energy efficiency of the considered  network. Then, we  formulated an optimization problem that maximizes the energy efficiency of the CR network operating in ABC  and HTT modes, for a given set of constraints including the primary interference constraint. Finally, we derived the optimal set of parameters that maximize the energy efficiency of the CR system. Capitalizing on the offered results, we quantified  the performance of the CR network under the considered  setup and demonstrated the performance improvement achieved in the CR network when incorporating a combination of ABC  and HTT modes. The offered results provided interesting theoretical and technical insights on the behavior of backscatter systems that are expected to be useful in the design and deployment of future systems in the context of various wireless applications of interest. 

\begin{appendices}
\section{Proof of Theorem \ref{etathm}} \label{etathmProof}
In order to establish that the constraint $P_{d} \geq \overline{P}_{d}$ is satisfied with equality, it is sufficient to  show that $ \partial \teeabc(\tau,\alpha,\mu,\varepsilon) / \partial \varepsilon  \geq 0$, for all $\varepsilon$. To this end, we observe that
\begin{align}
& \frac{\partial \teeabc(\tau,\alpha,\mu,\varepsilon)}{\partial \varepsilon} = \frac{\frac{\partial \rabc(\tau, \alpha,\mu, \varepsilon)}{\partial \varepsilon} \eabc(\tau, \alpha,\mu, \varepsilon)}{[\eabc(\tau, \alpha,\mu, \varepsilon)]^2} \nonumber \\
& ~~~~~~~~~~~~~~~~~~~~~ -\frac{\rabc(\tau, \alpha,\mu, \varepsilon)\frac{\partial \eabc(\tau, \alpha,\mu, \varepsilon)}{\partial \varepsilon}}{[\eabc(\tau, \alpha,\mu, \varepsilon)]^2}.
\label{firstdriEE}
\end{align}
Furthermore, taking the first derivative of \eqref{Ravg} with respect to $\epsilon$, namely  
\begin{align} \nonumber
& \rabc(\tau, \alpha,\mu, \varepsilon) = P(\mathcal{H}_{1})P_{d} (1 \hspace{-0.1cm} - \hspace{-0.1cm} \alpha) \tau B_{b} \hspace{-0.1cm} + \hspace{-0.1cm}  \kappa P(\mathcal{H}_{1})(1-P_{d}) \nonumber \\ 
& ~~~~~~~~~~~~~~ \mu \tau  W \hspace{-0.05cm} \log_{2} \hspace{-0.1cm} \left( \hspace{-0.1cm} 1 \hspace{-0.1cm} +  \hspace{-0.1cm} \frac{P_{tr}}{Z_{I}P_{T,PU}+ P_{0}} \hspace{-0.1cm}\right) \nonumber \\ 
& ~~~~~~~~~~~~~~  + P(\mathcal{H}_{0})(1-P_{f}) \mu \tau  W \hspace{-0.05cm} \log_{2} \hspace{-0.1cm} \left( \hspace{-0.1cm} 1 \hspace{-0.1cm} +  \hspace{-0.1cm} \frac{P_{tr}}{P_{0}} \hspace{-0.1cm}\right),
\end{align}
yields
\begin{align}
& \frac{\partial \rabc(\tau, \alpha,\mu, \varepsilon)}{\partial \varepsilon}=P(\mathcal{H}_{1})(1-\alpha)\tau B_{b}\frac{\partial P_{d}}{\partial \varepsilon} -\textcolor{black}{\kappa} \mu \tau W P(\mathcal{H}_{1}) \nonumber \\
& ~~~~~~~~~~~~~~ \log_{2} \hspace{-0.1cm} \left( \hspace{-0.1cm} 1 \hspace{-0.1cm} +  \hspace{-0.1cm} \frac{P_{tr}}{Z_{I}P_{T,PU}+ P_{0}} \hspace{-0.1cm}\right)\frac{\partial P_{d}}{\partial \varepsilon}  \nonumber \\
& ~~~~~~~~~~~~~~ - \mu \tau W P(\mathcal{H}_{0}) \log_{2} \hspace{-0.1cm} \left( \hspace{-0.1cm} 1 \hspace{-0.1cm} +  \hspace{-0.1cm} \frac{P_{tr}}{P_{0}} \hspace{-0.1cm}\right)\frac{\partial P_{f}}{\partial \varepsilon},
\label{firstdriR}
\end{align}
and
\begin{align}
& \frac{\partial \eabc(\tau, \alpha,\mu, \varepsilon)}{\partial \varepsilon}=-\mu \tau P_{tr} \left [P(\mathcal{H}_{1})\frac{\partial P_{d}}{\partial \epsilon}+P(\mathcal{H}_{0})\frac{\partial P_{f}}{\partial \epsilon}\right ].
\label{firstdriE}
\end{align}
Using the above, and substituting \eqref{firstdriR} and \eqref{firstdriE} in \eqref{firstdriEE}, and carrying out  some algebraic manipulations, one obtains 
\begin{align}
& \frac{\partial \eeabc(\tau,\alpha,\mu,\varepsilon)}{\partial \varepsilon} = \frac{\partial P_{d}}{\partial \varepsilon}\frac{P(\mathcal{H}_{1})(1-\alpha)\tau B_b}{\eabc(\tau, \alpha,\mu, \varepsilon)} \nonumber \\
& ~~~ -\frac{\partial P_{d}}{\partial \varepsilon} P(\mathcal{H}_{1})\textcolor{black}{\kappa}\mu \tau W \log_{2} \left ( 1+\frac{P_{tr}}{Z_{I}P_{T,PU}+ P_{0}}  \right ) \nonumber \\
& ~~~ + \frac{\partial P_{d}}{\partial \varepsilon}\mu \tau P_{tr}P(\mathcal{H}_{1}) \frac{\rabc(\tau, \alpha,\mu, \varepsilon)}{\eabc(\tau, \alpha,\mu, \varepsilon)^2} \nonumber \\
& ~~~ -\frac{\partial P_{f}}{\partial \varepsilon}\left [ P(\mathcal{H}_{0})\mu \tau W \log_{2} \left ( 1+\frac{P_{tr}}{ P_{0}}  \right )+\mu \tau P_{tr}P(\mathcal{H}_{0}) \right] \nonumber \\
& ~~~~~~~~~~~~~~ \left[\frac{\rabc(\tau, \alpha,\mu, \varepsilon)}{\eabc(\tau, \alpha,\mu, \varepsilon)^2}\right]. \label{RqdEqn1}
\end{align}
In addition, it is noted that 
\begin{align}
  \frac{\partial P_f}{\partial  \varepsilon}=-\exp {\left [ - \dfrac{(1-\tau)N_s(\frac{\epsilon}{\sigma^2}-1)^2}{2}\right ]} \frac{\sqrt{N_s(1-\tau)}}{\sqrt{2\pi \sigma^2}} \leq 0, \label{dpf} 
\end{align}
and
\begin{align}
 \frac{\partial P_d}{\partial  \varepsilon}=-\exp {\left [ - \dfrac{(1-\tau)N_s(\frac{\epsilon}{\sigma^2}-\gamma-1)^2}{2(1+2\gamma)}\right ]} \frac{\sqrt{\frac{N_s(1-\tau)}{1+2\gamma}}}{\sqrt{2\pi \sigma^2}} \leq 0, \label{dpd}
\end{align}
with $ \partial P_d/\partial  \varepsilon \geq  \partial P_f/\partial  \varepsilon$. Following these results, \eqref{RqdEqn1} can be further simplified and shown to be non-negative if
\begin{align}
& P(\mathcal{H}_{1}) \textcolor{black}{\kappa} \mu \tau W   \log_{2} \left ( 1+\frac{P_{tr}}{Z_{I}P_{T,PU}+ P_{0}}  \right ) \nonumber \\
&~ -\frac{1}{E}P(\mathcal{H}_{1}) (1-\alpha)\tau B_b - \mu \tau P_{tr}P(\mathcal{H}_{1})\frac{R}{E^2}>0.
\end{align}
It is easy to verify that the above requirement holds when $W$ and $P_{tr}$ are selected such that 
\begin{equation}
W \textcolor{black}{\kappa} \log\left(1+\frac{P_{tr}}{P_0}\right) \geq \frac{1}{E \mu}(1-\alpha) B_b +P_{tr}\frac{R}{E^2},
\end{equation}
in which case, $ \partial \teeabc / \partial \varepsilon  \geq 0$, for all $\varepsilon$. Based on this result, it is sufficient to choose the value of $\varepsilon$ when $P_{d}=\overline{P}_d$ is satisfied.

\section{Proof of Theorem \ref{formu}} \label{formuProof}
As mentioned earlier, we consider only the term $EE_{h}(\tau,\alpha,\mu, \varepsilon^*)$, since $EE_{b}(\tau,\alpha,\mu, \varepsilon^*)$ is independent of $\mu$. Next, it is noted that 
\begin{figure*}[ht]
\begin{align}
& EE_{h}(\tau,\alpha,\mu, \varepsilon^*) =    \frac{\kappa P(\mathcal{H}_{1})(1-P_{d}) \mu \tau  W \hspace{-0.05cm} \log_{2} \hspace{-0.1cm} \left( \hspace{-0.1cm} 1 \hspace{-0.1cm} +  \hspace{-0.1cm} \frac{P_{tr}}{Z_{I}P_{T,PU}+ P_{0}} \hspace{-0.1cm}\right)} {P_{s}(1 \hspace{-0.1cm} - \hspace{-0.1cm} \tau) \hspace{-0.1cm} + \hspace{-0.1cm} \mu \tau P_{t} \left \{1 \hspace{-0.1cm} - \hspace{-0.1cm} P(\mathcal{H}_{1}) P_{d} \hspace{-0.1cm} - \hspace{-0.1cm} P(\mathcal{H}_{0})P_{f}\right \}} \nonumber \\
& ~~~~~~~~~~~~~~~~~~~~~~~~~~~~~~~~~ + \frac{P(\mathcal{H}_{0})(1-P_{f}) \mu \tau  W \hspace{-0.05cm} \log_{2} \hspace{-0.1cm} \left( \hspace{-0.1cm} 1 \hspace{-0.1cm} +  \hspace{-0.1cm} \frac{P_{tr}}{P_{0}} \hspace{-0.1cm}\right) } {P_{s}(1 \hspace{-0.1cm} - \hspace{-0.1cm} \tau) \hspace{-0.1cm} + \hspace{-0.1cm} \mu \tau P_{t} \left \{1 \hspace{-0.1cm} - \hspace{-0.1cm} P(\mathcal{H}_{1}) P_{d} \hspace{-0.1cm} - \hspace{-0.1cm} P(\mathcal{H}_{0})P_{f}\right \}}. \label{Ehequ1}
\end{align}
\hrulefill
\end{figure*}
which upon substituting the expression for $P_{tr}$ from \eqref{ptrequation} yields \eqref{Ehequ2}.
\begin{figure*}[ht]
\begin{align}
& EE_{h}(\tau,\alpha,\mu, \varepsilon^*) =    \frac{\kappa P(\mathcal{H}_{1})(1-P_{d}) \mu \tau  W \hspace{-0.05cm} \log_{2}\left(1+\frac{\alpha \tau P_{R}-\mu \tau P_{c}-P_{s}(1-\tau)}{\left[Z_{I}P_{T,PU}+P_{0}\right]\tau \mu}\right)}{P_{s}(1 \hspace{-0.1cm} - \hspace{-0.1cm} \tau) \hspace{-0.1cm} + \hspace{-0.1cm} \mu \tau P_{t} \left \{1 \hspace{-0.1cm} - \hspace{-0.1cm} P(\mathcal{H}_{1}) P_{d} \hspace{-0.1cm} - \hspace{-0.1cm} P(\mathcal{H}_{0})P_{f}\right \}} \nonumber \\ 
& ~~~~~~~~~~~~~~~~~~~~~~~~~~~~~~~~~ + \frac{ P(\mathcal{H}_{0})(1-P_{f}) \mu \tau  W \hspace{-0.05cm} \log_{2}\left(1+\frac{\alpha \tau P_{R}-\mu \tau P_{c}-P_{s}(1-\tau)}{\left[P_{0}\right]\tau \mu}\right)}{P_{s}(1 \hspace{-0.1cm} - \hspace{-0.1cm} \tau) \hspace{-0.1cm} + \hspace{-0.1cm} \mu \tau P_{t} \left \{1 \hspace{-0.1cm} - \hspace{-0.1cm} P(\mathcal{H}_{1}) P_{d} \hspace{-0.1cm} - \hspace{-0.1cm} P(\mathcal{H}_{0})P_{f}\right \}}. \label{Ehequ2}
\end{align}
\hrulefill
\end{figure*}
\textcolor{black}{Compactly, \eqref{Ehequ2} can be re-written as}
\begin{align}
EE_{h}(\tau,\alpha,\mu, \varepsilon^*) =\frac{X_1 \mu\log_{2}\left[B \hspace{-0.1cm} + \hspace{-0.1cm} \frac{A}{\mu}\right]+X_2 \mu\log_{2}\left[D \hspace{-0.1cm} + \hspace{-0.1cm} \frac{C}{\mu}\right]}{X_3 + X_4 \mu},
\end{align}
where
\begin{align}
& X_1 \triangleq \textcolor{black}{\kappa} P(\mathcal{H}_{1})(1 \hspace{-0.1cm} - \hspace{-0.1cm} P_{d})\tau W, \\
& X_2 \triangleq P(\mathcal{H}_{0})(1 \hspace{-0.1cm} - \hspace{-0.1cm} P_{f})\tau W, \\
& X_3 \triangleq P_s(1-\tau), \\
& X_4 \triangleq \tau P_{tr}\{ P(\mathcal{H}_{1})(1-P_{d})+P(\mathcal{H}_{0})(1-P_{f}) \}, \\
& A\triangleq \frac{\alpha P_R-P_{s}(1 \hspace{-0.1cm} - \hspace{-0.1cm} \tau)}{[Z_{I}P_{T,PU}+ P_{0}]\tau}, \\
& B\triangleq 1-\frac{P_c}{[Z_{I}P_{T,PU}+ P_{0}]}, \\
& C \triangleq \frac{\alpha \tau P_R-P_s(1-\tau)}{P_{0}\tau},\end{align}
and
\begin{align} 
D \triangleq 1-\frac{ P_c}{P_{0}},
\end{align}
such that $X_1$, $X_2$, $X_3$, $X_4$, $A$, $B$, $C$ and $D$ are positive constants. 
Next, by evaluating the first derivative of $ EE_{h}(\tau,\alpha,\mu, \varepsilon^*)$ with respect to $\mu$ it follows that 
\begin{align}
& \frac{\partial EE_{h}(\tau,\alpha,\mu,\varepsilon^*)}{\partial \mu}\hspace{-0.1cm} =\hspace{-0.1cm} \frac{\frac{\partial R_{h}(\tau, \alpha,\mu, \varepsilon^*)}{\partial \mu} E_{h}(\tau, \alpha,\mu, \varepsilon^*)}{[E_{h}(\tau, \alpha,\mu, \varepsilon^*)]^2} \nonumber \\
& ~~~~~~~~~~~~  - \frac{R_{h}(\tau, \alpha,\mu, \varepsilon^*)\hspace{-0.1cm}\frac{\partial E_{h}(\tau, \alpha,\mu, \varepsilon^*)}{\partial \mu}}{[E_{h}(\tau, \alpha,\mu, \varepsilon^*)]^2}, \label{firstdriEEwrtomu}
\end{align}
where  
\begin{align}
&  \frac{\partial R_{h}(\tau, \alpha,\mu, \varepsilon^*)}{\partial \mu} \hspace{-0.1cm} = \hspace{-0.1cm} X_1\left \{ \hspace{-0.1cm} \frac{-A}{\mu(B+\frac{A}{\mu}) \ln (2)} \hspace{-0.1cm} + \hspace{-0.1cm} \log_2\left( \hspace{-0.05cm} B \hspace{-0.1cm} + \hspace{-0.1cm} \frac{A}{\mu} \hspace{-0.05cm} \right) \hspace{-0.1cm} \right \} \nonumber \\
& ~~~~~~~~~~~~  +\hspace{-0.1cm} X_2\left \{ \hspace{-0.1cm} \frac{-C}{\mu(D+\frac{C}{\mu}) \ln (2)} \hspace{-0.1cm} + \hspace{-0.1cm} \log_2\left( \hspace{-0.05cm} D \hspace{-0.1cm} + \hspace{-0.1cm} \frac{C}{\mu} \hspace{-0.05cm} \right) \hspace{-0.1cm} \right \}, \label{Rdiffwrtmu} 
\end{align}
and
\begin{align}
 \frac{\partial E_{h}(\tau, \alpha,\mu, \varepsilon^*)}{\partial \mu} \hspace{-0.1cm} = \hspace{-0.1cm} \tau P_{t} \left \{1  - P(\mathcal{H}_{1}) \overline{P}_{d} -  P(\mathcal{H}_{0})P_{f}\right \}. \label{Ediffwrtmu}
\end{align}
Substituting \eqref{Rdiffwrtmu} and \eqref{Ediffwrtmu} in \eqref{firstdriEEwrtomu}, we get
\begin{align}
& \frac{\partial EE_{h}(\tau,\alpha,\mu,\varepsilon^*)}{\partial \mu} = \frac{\hspace{-0.1cm} X_1\left \{ \hspace{-0.1cm} \frac{-A}{\mu(B+\frac{A}{\mu}) \ln (2)} \hspace{-0.1cm} + \hspace{-0.1cm} \log_2\left( \hspace{-0.05cm} B \hspace{-0.1cm} + \hspace{-0.1cm} \frac{A}{\mu} \hspace{-0.05cm} \right) \hspace{-0.1cm} \right \}}{\left ( X_3 + X_4 \mu\right )^2} \nonumber \\
& ~~~~~~~~~~~~ + \frac{X_2\left \{ \hspace{-0.1cm} \frac{-C}{\mu(D+\frac{C}{\mu}) \ln (2)} \hspace{-0.1cm} + \hspace{-0.1cm} \log_2\left( \hspace{-0.05cm} D \hspace{-0.1cm} + \hspace{-0.1cm} \frac{C}{\mu} \hspace{-0.05cm} \right) \hspace{-0.1cm} \right \}}{\left ( X_3 + X_4 \mu\right )^2}.
\end{align}
Based on this, we observe that  
\begin{align}
& \underset{\mu\rightarrow +\infty}{\lim}\frac{\partial EE_{h}(\tau,\alpha,\mu,\varepsilon^*)}{\partial \mu} \nonumber \\
& = \underset{\mu\rightarrow +\infty}{\lim}\frac{\frac{(X_1+X_2)X_3}{\mu}+X_4{(X_1+X_2)+X_1+X_2}}{\mu\left (\frac{ P_s(1-\tau)}{\mu} + X_4 \mu\right )^2}=0.
\end{align}
Moreover, the second derivative of $EE_{h}(\tau,\alpha,\mu,\varepsilon^*)$ with respect to $\mu$ can be expressed as follows: 
\begin{align}\label{secdirwrtmu}
 &\frac{\partial^2 EE_{h}(\tau,\alpha,\mu,\varepsilon^*)}{\partial \mu^2} \nonumber \\
 & ~~~~~~~ = -\frac{x_1 A}{E(A+B\mu)\ln2}+\frac{x_1\tau W \log_2\left (   \frac{A}{\mu}+B \right )}{E}\nonumber \\
 & ~~~~~~~~~ -\frac{x_2 C}{E(C+D\mu)\ln2}+\frac{x_2\tau W \log_2\left ( \frac{C}{\mu}+D \right )}{E}\nonumber \\
 & ~~~~~~~~~ -\frac{x_4\left \{ x_1 \mu \log(\frac{A}{\mu}+B)+x_2 \mu \log(\frac{C}{\mu}+D) \right \}}{\left [ x_3+\mu x_4 \right ]^2}.
\end{align}
Carrying out some long but straightforward   algebraic manipulations, it can be shown that $\frac{\partial^2 EE_{h}(\tau,\alpha,\mu,\varepsilon^*)}{\partial \mu^2} \leq 0$. Hence, $EE_{h}(\tau,\alpha,\mu,\varepsilon^*)$ is an increasing function of $\mu$, which implies that $\mu^*=1$.

\section{Proof of Theorem \ref{alpha}} \label{alphaProof}
Let us define
\begin{align}
& y_1 \triangleq 1- \frac{P_c \mu \tau }{  \left[Z_{I}P_{T,PU}+P_{0}\right]\tau \mu}, \label{eqny1} \\
& y_2 \triangleq  \frac{\tau P_{R}-P_s(1-\tau) }{  \left[Z_{I}P_{T,PU}+P_{0}\right]\tau \mu}, \label{eqny2} \\
& y_3 \triangleq 1-\frac{P_c }{P_0}-\frac{P_s(1-\tau) }{P_{0}\tau \mu}, \label{eqny3}
\end{align}
and
\begin{align}
& y_4 \triangleq \frac{\tau P_{R} }{P_{0}\tau \mu}, \label{eqny4}
\end{align}
such that $y_1$, $y_2$, $y_3$, and $y_4$ are positive constants. Then, the expression for energy efficiency is expressed as 
 \begin{align}
& \teeabc(\tau,\alpha,\mu^*,\varepsilon^*)= \hspace{-0.05cm} EE_{b}(\tau, \alpha,\mu^*,\varepsilon^*) \hspace{-0.1cm} + \hspace{-0.1cm} EE_{h}(\tau, \alpha,\mu^*,\varepsilon^*),\nonumber \\
& ~~ =\frac{R_b(\tau, \alpha,\mu^*,\varepsilon^*)}{E(\tau, \alpha,\mu^*,\varepsilon^*)}+\frac{R_h(\tau, \alpha,\mu^*,\varepsilon^*)}{E(\tau, \alpha,\mu^*,\varepsilon^*)} \nonumber \\
& ~~ = \frac{P(\mathcal{H}_{1})P_d(1-\alpha)\tau B_b}{E(\tau, \alpha,\mu^*,\varepsilon^*)} \nonumber \\
& ~~~~~ + \frac{P(\mathcal{H}_{1})(1-P_d)\textcolor{black}\kappa \mu \tau W \log_2\left [ y_1+\alpha y_2 \right ]}{E(\tau, \alpha,\mu^*,\varepsilon^*)} \nonumber \\
& ~~~~~ + \frac{P(\mathcal{H}_{0})(1-P_f) \mu \tau W \log_2\left [  y_3+\alpha y_4 \right ]}{E(\tau, \alpha,\mu^*,\varepsilon^*)}.
\end{align}
Now, consider the first derivative of $\eeabc(\tau, \alpha,  \mu^*, \varepsilon^*)$ with respect to $\alpha$, that is, $ \partial \eeabc(\tau,\alpha,\mu^*,\varepsilon^*)/\partial \alpha$, which is given by
\begin{align}
& \frac{\partial \eeabc(\tau,\alpha,\mu^*,\varepsilon^*)}{\partial \alpha}\hspace{-0.1cm}=\frac{P(\mathcal{H}_{1})(1-P_{d}) \textcolor{black}\kappa\frac{\mu \tau W} {\ln2}\frac{y_2} {y_1+\alpha y_2}}{  \eabc(\tau, \alpha,\mu^*, \varepsilon^*)} \nonumber \\
& ~~~ -\frac{P(\mathcal{H}_{1})P_d \tau B_b}{\eabc(\tau, \alpha,\mu^*, \varepsilon^*)} +\frac{P(\mathcal{H}_{0})(1-P_{f}) \frac{\mu \tau W} {\ln2}\frac{y_3} {y_1+\alpha y_4}}{  \eabc(\tau, \alpha,\mu^*, \varepsilon^*)}.\label{firstdirEE1}
\end{align}
Likewise, the second derivative of $\teeabc(\tau, \alpha,\mu^*, \varepsilon^*)$ is given by
\begin{align}
&\frac{\partial^2 \eeabc(\tau,\alpha,\mu^*,\varepsilon^*)}{\partial \alpha^2}\hspace{-0.1cm} =\hspace{-0.05cm}\frac{\partial^2 EE_{h}(\tau,\alpha,\mu^*,\varepsilon^*)}{\partial \alpha^2}, \nonumber \\
&~~~~~~ =-\frac{P(\mathcal{H}_{1}(1-P_d)\textcolor{black}\kappa}{E} \frac{\mu \tau W}{\ln2}\frac{y_{2}^2}{(y_1+\alpha y_2)^2} \nonumber \\
&~~~~~~~~~ -\frac{P(\mathcal{H}_{0})(1-P_f)}{E} \frac{\mu \tau W}{\ln2}\frac{y_{4}^2}{(y_1+\alpha y_2)^2} <0 \label{secdrivEEbandEEh2}
\end{align}
From  \eqref{firstdirEE1} and \eqref{secdrivEEbandEEh2}, we can infer that $ \partial \eeabc(\tau,\alpha,\mu^*,\varepsilon^*)/\partial \alpha$ is a decreasing function of $\alpha$. Furthermore, to guarantee that there exist a value of $\alpha \in [\alpha^{\dagger},1]$ such that $EE^{'}_{ABC}(\tau, \alpha,  \mu^*, \varepsilon^*)=0$, we calculate the following boundary  values. To this effect, observing that when $\alpha = \alpha^\dagger$,  it follows that
\begin{align}
& \dfrac{\partial \eeabc(\tau, \alpha^\dagger,  \mu^*, \varepsilon^*)}{\partial \alpha} =\frac{P(\mathcal{H}_{1})(1-P_{d})\textcolor{black}\kappa \frac{\mu \tau W} {\ln2}\frac{y_2} {y_1+\alpha^\dagger y_2}}{  \eabc(\tau, \alpha,\mu^*, \varepsilon^*)}\nonumber \\
& ~ -\frac{P(\mathcal{H}_{1})P_d \tau B_b}{\eabc(\tau, \alpha,\mu^*, \varepsilon^*)} +\frac{P(\mathcal{H}_{0})(1-P_{f}) \frac{\mu \tau W} {\ln2}\frac{y_3} {y_1+\alpha^\dagger y_4}}{  \eabc(\tau, \alpha,\mu^*, \varepsilon^*)}~ \geq 0, \label{ForBbubEqn}
\end{align}
whereas  when $\alpha = 1$, we get 
\begin{align}\label{ForBblbeqn}
& \frac{\partial \eeabc(\tau, 1,  \mu^*, \varepsilon^*)}{\partial \alpha}= \frac{P(\mathcal{H}_{1})(1-P_{d})\textcolor{black}\kappa \frac{\mu \tau W} {\ln2}\frac{y_2} {y_1+ y_2}}{  \eabc(\tau, \alpha,\mu^*, \varepsilon^*)}\nonumber \\
& ~ -\frac{P(\mathcal{H}_{1})P_d \tau B_b}{\eabc(\tau, \alpha,\mu^*, \varepsilon^*)} +\frac{P(\mathcal{H}_{0})(1-P_{f}) \frac{\mu \tau W} {\ln2}\frac{y_3} {y_1+ y_4}}{  \eabc(\tau, \alpha,\mu^*, \varepsilon^*)}~ \leq 0.
\end{align} 
Therefore, there exists an $\alpha^* \in [\alpha^\dagger, 1]$ where the derivative $\dfrac{\partial \eeabc(\tau, \alpha^\dagger,  \mu^*, \varepsilon^*)}{\partial \alpha}$ is exactly $0$. Thus, the bounds on $B_b$ are obtained by equating the expression in \eqref{firstdirEE1} to zero, and rearranging accordingly, yielding Eq.\eqref{Bbeqn}.
\begin{figure*}[ht]
\begin{align}
& B_b=\frac{1-P_d}{P_d}\frac{\textcolor{black}\kappa \mu \tau W}{\ln2}\frac{\tau P_R-P_s(1-\tau)}{\left [ Z_{I}P_{T,PU}+P_{0}-P_c\right ]\tau \mu+\alpha(\tau P_R-P_s(1-\tau))} \nonumber \\
& ~~~~~~~~~~~~  + \frac{P(\mathcal{H}_{0})}{P(\mathcal{H}_{1})}\frac{(1-P_f)}{P_d}\frac{\mu \tau W}{\ln2}\frac{\tau P_R}{(P_0 -P_c )\mu \tau +P_s(1-\tau)+\alpha \tau P_R}. \label{Bbeqn}
\end{align}
\hrulefill
\begin{align}
& \frac{\partial EE_{b}(\tau,\alpha^*,\mu^*,\varepsilon^*)}{\partial \tau}\hspace{-0.05cm}=\hspace{-0.1cm}\frac{P(\mathcal{H}_{1})\overline{P_d}B_{b}\left \{ 1-\alpha-\tau \frac{\partial \alpha^*}{\partial \tau}
\right \}}{\left[\eabc(\tau, \alpha^*,\mu^*, \varepsilon^*) \right]} \nonumber \\
& ~~~ + \frac{\left \{ P(\mathcal{H}_{1})\overline{P_d}(1-\alpha)\tau B_{b} \right \}\left \{ -P_s+\mu P_{tr}\left [ P(\mathcal{H}_{1})(1-P_d)\textcolor{black}\kappa-P(\mathcal{H}_{0})\tau \frac{\partial P_f}{\partial \tau}+P(\mathcal{H}_{0})(1-P_f) \right ] \right \}  }{\left[\eabc(\tau, \alpha^*,\mu^*, \varepsilon^*) \right]^2}. \tag{82} \label{EEb} 
\end{align}
\hrulefill
\begin{align}
&\frac{\partial  EE_{h}(\tau,\alpha^*,\mu^*,\varepsilon^*)}{\partial \tau} = \left \{ \frac{1}{\eabc(\tau, \alpha^*,\mu^*, \varepsilon^*)} \right \}P(\mathcal{H}_{1})(1-\overline{P_d})\textcolor{black}\kappa\mu W \left \{ \frac{w_2^{2}}{\tau w_1-w_2} +\log(w_1-\frac{w_2}{\tau})\right \}\nonumber \\ &-\left \{ \frac{1}{\eabc(\tau, \alpha^*,\mu^*, \varepsilon^*)} \right \}\frac{(1-P_f)w_4w_3}{\ln2(\tau w_3-w_4)^2}- \left \{ \frac{1}{\eabc(\tau, \alpha^*,\mu^*, \varepsilon^*} \right \}\frac{\frac{\partial P_f}{\partial \tau} w_4}{\ln2 (\tau w_3-w_4)}\nonumber \\ &-\left \{ \frac{1}{\eabc(\tau, \alpha^*,\mu^*, \varepsilon^*)} \right \}\frac{\partial P_f}{\partial \tau}\log(w_3-\frac{w_4}{\tau})  +\left \{ \frac{1}{\eabc(\tau, \alpha^*,\mu^*, \varepsilon^*)} \right \}\frac{(1-P_f)w_4}{ \tau (\tau w_3-w_4)\tau \ln2} \nonumber \\&-\frac{1}{\eabc(\tau, \alpha^*,\mu^*, \varepsilon^*)^2}\left \{ -P_s+\mu P_{tr}\left [ P(\mathcal{H}_{1})(1-P_d)\textcolor{black}\kappa-P(\mathcal{H}_{0})\tau \frac{\partial P_f}{\partial \tau}+P(\mathcal{H}_{0})(1-P_f) \right ] \right \}\nonumber \\& \left \{  P(\mathcal{H}_{1})(1-P_d)\textcolor{black}\kappa\mu \tau W \log_2\left [w_1-\frac{w_2}{\tau}\right ]  +P(\mathcal{H}_{0})(1-P_f) \mu \tau W \log_2\left [ w_3-\frac{w_4}{\tau} \right ]\right \}. \tag{83}
\label{EEh}
\end{align}
\hrulefill
\end{figure*}
If the effective interference from the PU is neglected, then
\begin{align}
& B_b=\frac{P(\mathcal{H}_{0})}{P(\mathcal{H}_{1})}\frac{(1-P_f)}{P_d}\frac{\mu \tau W}{\ln2} \nonumber \\
& ~~~~~ \times \frac{\tau P_R}{(P_0 -P_c )\mu \tau +P_s(1-\tau)+\alpha \tau P_R}.
\end{align}
Now, the upper and lower bounds on $B_b$, namely, $\Bblb$ and $\Bbub$ can be obtained by substituting for the value of $\alpha$ corresponding to the two extreme cases $1$ and $\alpha^\dagger$, respectively. These bounds are given by
\begin{align}
& \Bblb \triangleq \frac{P(\mathcal{H}_{0})}{P(\mathcal{H}_{1})}\frac{(1-P_f)}{P_d}\frac{\mu \tau W}{\ln2}  \\
& ~~~~~ \times \frac{\tau P_R}{(P_0 -P_c )\mu \tau +P_s(1-\tau)+ \tau P_R}, \nonumber \\
& \Bbub \triangleq \frac{P(\mathcal{H}_{0})}{P(\mathcal{H}_{1})}\frac{(1-P_f)}{P_d}\frac{\mu \tau W}{\ln2} \nonumber \\
& ~~~~~ \times \frac{\tau P_R}{(P_0 -P_c )\mu \tau +P_s(1-\tau)+\alpha^{\dagger} \tau P_R}.
\end{align}
Therefore, when $B_b \in (\Bblb, \Bbub)$, $\teeabc(\tau, \alpha,  \mu^*, \varepsilon^*)$ is concave in $\alpha$. Finally, the optimal $\alpha^* \in [\alpha^{\dagger},1]$, can be obtained by neglecting the interference term and equating the first derivative to zero, which is expressed  as
\begin{align}
& \alpha^*=\frac{P(\mathcal{H}_{0}) }{P(\mathcal{H}_{1})}\frac{(1-P_f)}{P_d}\frac{\mu \tau W}{\ln 2}-\frac{y_3}{y_4},
\end{align}
for $B_b \in (\Bblb, \Bbub)$. Based on this, by substituting for $y_3$ and $y_4$, we obtain
\begin{align}
& \alpha^*=\frac{P(\mathcal{H}_{0}) }{P(\mathcal{H}_{1})}\frac{(1-P_f)}{P_d}\frac{\mu \tau W}{\ln 2} \nonumber \\
& ~~~~~~~~ -\frac{(P_{0} +P_{c})\tau \mu+P_s(1-\tau)}{\tau P_R}.
\end{align}

\section{Proof of Theorem \ref{fortau}} \label{fortauProof}
Let us define 
\begin{align}
& w_1 \triangleq 1+ \frac{\alpha P_R-\mu P_c+P_s }{  \left[Z_{I}P_{T,PU}+P_{0}\right]\mu}, \\
& w_2 \triangleq  \frac{P_s }{  \left[Z_{I}P_{T,PU}+P_{0}\right] \mu}, \\
& w_3 \triangleq 1+\frac{\alpha P_R-\mu P_c+P_s }{P_{0} \mu},
\end{align}
and
\begin{align}
& w_4 \triangleq \frac{P_{s} }{P_{0} \mu},
\end{align}
such that $w_1$, $w_2$, $w_3$, and $w_4$ are positive constants. Then, the expression for energy efficiency can be simplified as
\begin{equation}
 \begin{split}
 \teeabc(\tau,\alpha^*,\mu^*,\varepsilon^*) =& \frac{P(\mathcal{H}_{1})P_d(1-\alpha^*)\tau B_b}{E(\tau, \alpha,\mu^*,\varepsilon^*)}  \\
 & +\frac{P(\mathcal{H}_{1})(1-P_d)\textcolor{black}\kappa\mu \tau W \log_2\left [w_1-\frac{w_2}{\tau}\right ]}{E(\tau, \alpha,\mu^*,\varepsilon^*)} \\
 & +\frac{P(\mathcal{H}_{0})(1-P_f) \mu \tau W \log_2\left [ w_3-\frac{w_4}{\tau} \right ]}{E(\tau, \alpha,\mu^*,\varepsilon^*)}.
 \end{split}
\end{equation}
Next, we determine  the first derivative of $\eeabc(\tau, \alpha^*,  \mu^*, \varepsilon^*))$ with respect to $\tau$ i.e. $ \partial \eeabc(\tau,\alpha^*,\mu^*,\varepsilon^*)/\partial \tau $, upon which it is clear that
\begin{align}
& \frac{\partial \teeabc(\tau,\alpha^*,\mu^*,\varepsilon^*)}{\partial \tau}\hspace{-0.1cm}=\frac{\partial EE_{b}(\tau,\alpha^*,\mu^*,\varepsilon^*)}{\partial \tau}\hspace{-0.05cm} \nonumber \\
& ~~~~~~~~~~~~~~~~~~~ +\hspace{-0.05cm}\frac{\partial EE_{h}(\tau,\alpha^*,\mu^*,\varepsilon^*)}{\partial \tau}. \label{EEbandEEh1}
\end{align}
By rewriting $\alpha^*$ in \eqref{alphaopt} in terms of $\tau$ as
\begin{align}
\alpha^*=L_1\frac{(1-P_f)}{P_d}\tau-L_2-L_3\left(\frac{1}{\tau}-1\right), \label{alphaopt2}
\end{align}
where 
\begin{align}
& L_1 \triangleq \frac{P(H_0)}{P(H_1)}\frac{\mu W}{ln2} \geq 0, \\
& L_2 \triangleq \frac{(P_0+P_c)\mu}{P_R} \geq 0,
\end{align}
and 
\begin{align}
& L_3 \triangleq \frac{P_s}{P_R} \geq 0,
\end{align}
the derivative of $\alpha^*$ can be calculated as
\begin{align}
\frac{\partial \alpha^*}{\partial \tau}=-\frac{L_1P_f'}{P_d}-\frac{L_1(1-P_f)P_d'}{P_d^2}+\frac{L_3}{\tau^2}. \label{alphaoptdir}
\end{align}
Substituting \eqref{alphaopt2} and \eqref{alphaoptdir} in \eqref{EEbandEEh1}, and observing that $P_d=\overline{P_d}$ at $\varepsilon = \varepsilon^*$,we obtain $ \partial EE_{b}(\tau,\alpha^*,\mu^*,\varepsilon^*)/\partial \tau$ and $ \partial EE_{h}(\tau,\alpha^*,\mu^*,\varepsilon^*)/\partial \tau$, as given in \eqref{EEb} and \eqref{EEh}, respectively.
In addition, it is straightforward to calculate the first derivative of $P_f$ with respect to $\tau$ as
\begin{align}
& \frac{\partial P_f}{\partial \tau} = \dfrac{\left(\frac{\varepsilon^*}{\sigma^2} \hspace{-0.1cm} - \hspace{-0.1cm} 1\right)}{\sqrt{8 \pi^2 N_s(1 \hspace{-0.1cm} - \hspace{-0.1cm} \tau)}} \exp\left(-\frac{N_s(1 \hspace{-0.1cm} - \hspace{-0.1cm} \tau)}{2} \left[\frac{\varepsilon^*}{\sigma^2} \hspace{-0.1cm} - \hspace{-0.1cm} 1\right]^2\right),
\end{align}
and it is intuitive that the second derivative of $P_f$ with respect to $\tau$ would be negative, since $P_f$ is concave in $\tau$. Utilizing this result, the rest of the proof involves calculation of the second derivatives of $EE_{b}(\tau,\alpha^*,\mu^*,\varepsilon^*)$ and $EE_{h}(\tau,\alpha^*,\mu^*,\varepsilon^*)$, and showing them to be negative. Therefore, the second derivative of $\eeabc(\tau,\alpha^*,\mu^*,\varepsilon^*)$ is also negative, and hence $\eeabc(\tau,\alpha^*,\mu^*,\varepsilon^*)$ is concave in $\tau$. The corresponding details lead to expressions that are rather lengthy, which are omitted for brevity.
\end{appendices}

\bibliographystyle{IEEEtran}                              
\bibliography{IEEEabrv,abcbiblo}

\end{document}